\documentclass[a4paper,UKenglish,cleveref, autoref, thm-restate,noalgorithms]{lipics-v2021}
\hideLIPIcs
\nolinenumbers
\supplement{}
\supplementdetails[subcategory={}, cite={}, swhid={}]{Dataset}{https://doi.org/10.5281/zenodo.15778413}

\supplementdetails[subcategory={}, cite={}, swhid={}]{Software}{https://github.com/HeiHGM/Streaming}

\usepackage{hyperref}
\usepackage{tcolorbox}
\usepackage{microtype}
\usepackage{xcolor} 
\usepackage{todonotes}
\usepackage{xspace}
\usepackage{algorithm}
\usepackage{algorithmic}
\usepackage[dot]{dashundergaps}
\usepackage[binary-units=true, detect-all]{siunitx}
\DeclareSIUnit{\noop}{\relax}
\usepackage{booktabs}
\def\Vhrulefill{\leavevmode\leaders\hrule height 0.7ex depth \dimexpr0.4pt-0.7ex\hfill\kern0pt}

\usepackage{balance} 
\usepackage{xspace}
\usepackage{xargs}  
\newcommandx{\smf}[2][1=]{\todo[color=cyan!50,#1]{\sf \textbf{smf:} {\footnotesize #2}}\xspace}
\newcommandx{\AP}[2][1=]{\todo[color=magenta!50,#1]{\sf \textbf{Alex:} {\footnotesize #2}}\xspace}
\newcommand\algorithmicprocedure{\textbf{procedure}}
\newcommand{\algorithmicendprocedure}{\algorithmicend\ \algorithmicprocedure}
\makeatletter

\newcommand\PROCEDURE[3][default]{%
  \ALC@it
  \algorithmicprocedure\ \textsc{#2}(#3)%
  \ALC@com{#1}%
  \begin{ALC@prc}%
}
\newcommand\ENDPROCEDURE{%
  \end{ALC@prc}%
  \ifthenelse{\boolean{ALC@noend}}{}{%
    \ALC@it\algorithmicendprocedure
  }%
}
\newenvironment{ALC@prc}{\begin{ALC@g}}{\end{ALC@g}}
\makeatother

\keywords{hypergraph, matching, semi-streaming}
\ccsdesc[300]{Mathematics of computing~Graph algorithms}
\ccsdesc[500]{Mathematics of computing~Hypergraphs} %
\newcommand{\weight}{W}
\newcommand{\bestEdgeSymbol}{\mathcal{B}}
\newcommand{\cpp}{\textsf{C}\texttt{++}\xspace}
\authorrunning{Reinstädtler et al.}
\author{Henrik Reinstädtler}{Heidelberg University, Germany}{henrik.reinstaedtler@informatik.uni-heidelberg.de}{https://orcid.org/0009-0003-4245-0966}{}
\author{S M Ferdous}{Pacific Northwest National Laboratory, Richland, WA, USA}{sm.ferdous@pnnl.gov}{https://orcid.org/0000-0001-5078-0031}{}
\author{Alex Pothen}{Purdue University, West Lafayette, IN, USA}{apothen@purdue.edu}{https://orcid.org/0000-0002-3421-3325}{}
\author{Bora Uçar}{CNRS and LIP, ENS de Lyon, France \and UMR5668 (CNRS, ENS de Lyon, Inria, UCBL1) France \and 
\url{http://perso.ens-lyon.fr/bora.ucar}}{bora.ucar@ens-lyon.fr}{https://orcid.org/0000-0002-4960-3545}{}
\author{Christian Schulz}{Heidelberg University, Germany}{christian.schulz@informatik.uni-heidelberg.de}{https://orcid.org/0000-0002-2823-3506}{}
\funding{We acknowledge support by DFG grant \hbox{SCHU 2567/8-1}. Moreover, we like to acknowledge Dagstuhl Seminar 24201 on discrete algorithms on modern and emerging compute infrastructure. 
Alex Pothen's research was supported by the U.S. Department of Energy grant  SC-0022260.
 S M Ferdous's research was supported by the Laboratory Directed Research and Development Program at PNNL.}

\title{Semi-Streaming  Algorithms for Hypergraph Matching}

\date{}
\begin{document}

\maketitle

\begin{abstract} %

We propose two one-pass streaming algorithms for the $\mathcal{NP}$-hard hypergraph matching problem. The first algorithm stores a small subset of potential matching edges in a stack using dual variables to select edges. It has an approximation guarantee of $\frac{1}{d(1+\varepsilon)}$ and requires $\mathcal{O}((\frac{n}{\varepsilon}) \log^2{n})$ bits of memory, where $n$ is the number of vertices in the hypergraph,  $d$ is the maximum number of vertices in a hyperedge, and $\epsilon > 0$ is a parameter to be chosen.
The second algorithm computes, stores, and updates a single matching as the edges stream, with an approximation ratio dependent on a parameter $\alpha$. Its best approximation guarantee is $\frac{1}{(2d-1) + 2 \sqrt{d(d-1)}}$, and it requires only $\mathcal{O}(n)$ memory.

We have implemented both algorithms and compared them with respect to solution quality, memory consumption, and running times on two diverse sets of hypergraphs with a non-streaming greedy and a naive streaming algorithm.
 Our results show that the streaming algorithms achieve much better solution quality than naive algorithms when facing adverse orderings.
 Furthermore, these algorithms reduce the memory required by  a factor of 13 in the geometric mean on our test problems, and also outperform the offline Greedy algorithm in running time.

\end{abstract}

\clearpage

\section{Introduction}
We propose two streaming algorithms for the hypergraph matching problem, derive their approximation ratios, and implement them to evaluate their practical performance on a large number of test hypergraphs. 
Recall that hypergraphs are a natural extension of graphs and can help to model our ever evolving earth and society. 
In a hypergraph, a hyperedge is a subset of vertices and can contain any number of them, instead of just two.
The hypergraph matching problem asks for a set of vertex-disjoint hyperedges.
Two common objectives in the hypergraph matching problem are to maximize the number or total weight of the matching hyperedges.
The hypergraph matching problem has applications ranging from personnel scheduling~\cite{froger2015set} to resource allocations in combinatorial auctions~\cite{gottlob2013decomposing}.
The hypergraph matching problem with either of the objective \hbox{functions is~$\mathcal{NP}$-complete~\cite{approxresult}}, but a simple Greedy algorithm is $1/d$-approximate, where $d$ is the maximum size of a hyperedge~\cite{Besser+:Greedy,Korte+:Greedy}. 

There are a few papers with computational studies for the hypergraph matching problem, in a single CPU in-memory setting~\cite{dufosse2019effective} and in a distributed computing setting~\cite{hanguir2021distributed}.
However, %
little to no attention was paid to increasing 
data sizes and approximation guarantees.
Streaming and semi-streaming algorithms address this trend of ever-increasing data size.
In a streaming setting, hyperedges arrive one by one in arbitrary order. %
The amount of memory that can be used is strictly bounded by the size of the final solution.
In the case of hypergraph matching, this is~$\Theta(n)$,  where $n$ is the number of vertices in the hypergraph, because every vertex can be matched at most once.
For a semi-streaming setting, this criterion is relaxed to allow for an additional polylog factor.
Furthermore, establishing a bound on the degree of suboptimality is essential for evaluating the solution's effectiveness. 
For matchings in graphs, there is a semi-streaming algorithm having an approximation~guarantee of~$\frac{1}{2}$ \hbox{by Paz~and~Schwartzman~\cite{10.1145/3274668}.}

The algorithm by Paz~and~Schwartzman requires one variable per vertex, the %
dual variable, and uses a stack to store candidate matching edges.
 When an edge appears in the stream, the algorithm adds it to a stack if its weight dominates the sum of the duals of its vertices, and then updates the duals with the difference between the edge weight and the sum of duals. Otherwise, the edge is discarded.
After all edges have been streamed, non-conflicting edges are added from the stack beginning at the top. %
Ghaffari and Wajc~\cite{ghaffari2017simplified} give a simplified proof of the approximation guarantee of this algorithm 
using the primal-dual linear programming framework.

Our contributions are as follows.
We first propose a novel streaming framework for hypergraph matching and prove an approximation guarantee in relation to the largest hyperedge size
by extending the stack-based algorithm of Paz and Schwartzman~\cite{10.1145/3274668} to hypergraphs. 
In essence, our algorithm puts hyperedges that potentially belong to a good matching on a stack, and in the end computes a matching out of those hyperedges.
 Given the maximum hyperedge size~$d$ (the maximum number of vertices in a hyperedge), we use primal-dual techniques to prove a~$\frac{1}{d(1+\varepsilon)}$ approximation factor, where $\epsilon>0$ is  a parameter to be chosen.
 Our most memory-saving algorithm requires $\mathcal{O}(n\log^2{n}/\varepsilon)$ bits of space.
 We then propose a second family of algorithms which do not need a stack and require less space and work by greedily swapping hyperedges from the current matching with the incoming ones. 
 Like the first stack-based algorithm, this algorithm family requires~$\mathcal{O}(\left|E\right|\cdot d)$ work.
 They have an approximation guarantee depending on a factor $\alpha>0$, which can be tuned to result in a guarantee of $\frac{1}{(2d-1)+2\sqrt{(d-1)d}}$.
In experiments, we show the competitiveness of our approaches and benchmark them on a set of social link hypergraphs and a large set of instances from hypergraph partitioning.
 We compare our algorithms with a \textsf{Naive} streaming algorithm that maintains a maximal matching in the hypergraph by adding a hyperedge from the stream to the match if it does not overlap in any vertex with the current matching edges, and a non-streaming \textsf{Greedy} algorithm. 
 The stack-based algorithms reduce the memory consumption by up to~13 times in comparison to the non-streaming \textsf{Greedy} algorithm on social link hypergraphs.
We investigate the impact of ordering the hyperedges in the stream and show that our stack algorithm can handle them better than the non-streaming \textsf{Greedy} and \textsf{Naive} streaming algorithms, while requiring only 26\% more time than the \textsf{Naive} algorithm. 
We show the impact of the parameters $\alpha$ and $\varepsilon$ on their respective algorithms.

The rest of the paper is organized as follows.
After introducing the notation and related work in Section~\ref{sec:prelim}, we show our approximation guarantee for an adaptation of the Paz-Schwartzman semi-streaming algorithm and discuss further improvements in Section~\ref{sec:quality}.
Section~\ref{sec:stackless} introduces our greedy swapping algorithm for streaming based on McGregors~\cite{mcgregor2005finding} algorithm.
These approaches are then extensively evaluated by experiments in Section~\ref{sec:experiments}. We conclude in Section~\ref{sec:conclusion}.
The Appendix~\ref{sec:stackless:quality} of the paper contains 
the proof of the approximation guarantee for the algorithm from Section~\ref{sec:stackless}, and additional statistics on the test problems as well in Appendix~\ref{app:stats}.
\section{Preliminaries}\label{sec:prelim}
\subsection{Basic Concepts}

\subparagraph{Hypergraphs.}A \emph{weighted undirected hypergraph}~$H=(V,E,\weight)$ consists of a set~$V$ of~$n$ vertices and a set $E$ of~$m$ hyperedges.
 Each hyperedge~$e$ is a subset of vertices and is assigned a positive weight by the weight function~$\weight\colon E\to \mathbb{R}_{>0}$.
 The number of vertices in a hyperedge~$e$ is its size, and is denoted by~$\left|e\right|$, and the maximum size of a hyperedge or rank of the hypergraph is denoted by~$d:=\max_{e\in E}{\left|e\right|}$.
 For clarity and brevity, we refer to a hyperedge simply as an edge
when it is evident from the context that a hypergraph is under
consideration.
  
 \subparagraph{Matching.} A subset of (hyper-)edges~$M\subset E$ is a \emph{matching}, if all (hyper-)edges in $M$ are pairwise disjoint, i.e.,~only at most one (hyper-)edge is selected at every vertex.
  A matching $M$ is called maximal, if there is no (hyper-)edge in~$E$ which can be added to $M$ without violating the matching constraint. The weight of a matching is defined by~$\weight(M):=\sum_{e\in M}\weight(e)$ and a maximum matching is a \hbox{matching with the largest weight.}

  \subparagraph{Related $\mathcal{NP}$-hard Problems. }
  The unweighted hypergraph matching problem is closely related to the maximum independent set and the~$k$-set packing problems.
  Both problems are $\mathcal{NP}$-hard~\cite{karp2010reducibility}.
 An independent set in a graph is a subset of vertices in which no two vertices are adjacent.
 There is a simple transformation from hypergraph matching to maximum independent set using the line graph of the hypergraph; in the line graph, every hyperedge becomes a vertex, and two such vertices are connected if the corresponding hyperedges share a common vertex.
Given a ground set~$S$ and some subsets~$S_1,\dots,S_n$, each of size at most~$k$, the~$k$-set packing problem asks to select the maximum number of disjoint subsets.
It can be translated to the hypergraph matching setting, by choosing the set~$S$ to correspond to the vertices~$V$, while the subsets~$S_1,\dots,S_n$ correspond to the~hyperedges. 

 \subparagraph{(Semi-)Streaming Algorithms.}
 If the input size exceeds the memory of a machine, a typical solution is to stream the input.
 There are several definitions for streaming in graphs and hypergraphs.
 When (semi-)streaming, the (hyper-)edges of a (hyper-)graph are usually presented in an arbitrary (even adverse) order one-by-one in several passes.
 In this paper we consider having only one pass over the input.
 In a streaming setting, the memory is strictly bounded by the solution size and, hence, for matching in hypergraphs, it is limited by the number of vertices.
 When using the semi-streaming model the memory is %
 bounded by~$\mathcal{O}(n\cdot\mathrm{polylog}(n))$.

 \subparagraph{Approximation Factors.}
 Algorithms can be classified into three categories: exact algorithms, heuristics without approximation guarantees, and approximation algorithms. The quality of an approximation algorithm is measured by comparing its solution's value to that of an optimal solution.
 For an instance $I$ of a maximization problem, let the optimum objective be denoted by $M(I)$. If an algorithm~$\mathcal{A}$ is guaranteed to find a solution that is bigger than $
\alpha M(I)$ for every instance,  where $\alpha \in \mathbb{R}^+_{<1}$ and is chosen to be as large as possible,  then $\mathcal{A}$ provides an $\alpha$-approximation guarantee. 
In some communities, the convention is to report $\frac{1}{\alpha}>1$ as the approximation~ratio, although we do not follow it here.
 For an overview of techniques to design approximation algorithms, we refer the reader to~Williamson~and~Shmoys~\cite{williamson2011design}.
 \subparagraph{(Integer) Linear Programs.}
 Many optimization problems can be formulated as integer linear programs (ILP). 
 In a maximization problem, an (integer) linear program finds an (integer) vector $x$ with components~$x_i$, that maximizes a linear cost function~$\sum c_ix_i$ such that a constraint~$Ax\leq b$, where~$A$ is a matrix, is satisfied, with typically additional constraints on the components of $x$, e.g., $x_i\geq 0$.
 If the variables of the problem are integer, some problems are~$\mathcal{NP}$-hard, while other problems 
 (where the constraint matrix is unimodular) are solvable in polynomial time~\cite{wolsey2020integer}.
 When the integer constraint is dropped, any linear program is solvable in polynomial time~\cite{wright1997primal}.
For every linear program one can find a dual problem~\cite{williamson2011design}.
 Given a maximization problem as described before, the dual problem is to find a vector $y_i\geq 0$, that minimizes~$\sum b_iy_i$ subject to~$A^Ty\geq c$.
The weak duality theorem~\cite{williamson2011design} states that for any primal maximization problem, the dual minimization problem for any feasible solution has an objective value
larger than the optimal solution of the primal problem.
The strong duality theorem states that if the primal problem has an optimal solution then the dual is solvable as well and the optimum values are the same.
For a more detailed introduction, we refer the reader to the Appendix~of~\cite{williamson2011design}.

\subsection{Related Work}
Matching is a well-studied problem in computer science, and here
 we give a brief overview of matchings in graphs and hypergraphs.
\subparagraph{Matching in Graphs and Streaming.}
The polynomial-time complexity of matchings in graphs is one of the classical results in theoretical computer science~\cite{edmonds_1965}.
While Preis~\cite{preis1999linear} presents the first linear time~$\frac{1}{2}$-approximation, Drake and Hougardy~\cite{drake2003simple} show a simpler algorithm with the same approximation ratio by path growing (PGA) in linear time. 
A number of other $1/2$-approximation algorithms have been developed, including the proposal-based Suitor algorithm of Manne and Halappanavar~\cite{DBLP:conf/ipps/ManneH14}.
 Pettie and Sanders~\cite{PETTIE2004271} propose a~$\frac{2}{3}-\varepsilon$ approximation with expected running time of~$\mathcal{O}(m\log \frac{1}{\varepsilon})$.
 The GPA algorithm by Maue and Sanders~\cite{maue2007engineering} bridges the gap between greedy and path-searching algorithms, showing that a combination of both works best in practice.
 Pothen, Ferdous and Manne~\cite{Pothen+:survey} survey these approximation algorithms. 
 Birn~et~al.~\cite{birn2013efficient} develop a parallel algorithm in the CREW PRAM model with $\frac{1}{2}$-approximation guarantee and $\mathcal{O}(\log^2n)$ work.
Feigenbaum~et~al.~\cite{feigenbaum2005graph} present a $\frac{1}{6}$-approximation for the weighted matching problem in the semi-streaming setting using a blaming-based analysis.
McGregor~\cite{mcgregor2005finding} develop a multipass algorithm that returns a $\frac{1}{2+\varepsilon}$ approximation in $\mathcal{O}(\varepsilon^{-3})$ rounds, with the initial matching having an approximation guarantee of $\frac{1}{3 + 2\sqrt{2}}$.
Paz and Schwartzman~\cite{10.1145/3274668} give a~$\frac{1}{2+\varepsilon}$-approximation algorithm, 
which employs a dual solution to admit candidate edges into a stack, while also 
updating the dual solution.
 The matching is constructed by removing edges from the top of the stack, and those that do not violate the matching property are added to the solution.
 The resulting matching is not necessarily maximal.
 Ghaffari and Wajc~\cite{ghaffari2017simplified} provide a simpler proof of the approximation ratio using a primal-dual analysis. 
 Ferdous~et~al.~\cite{ferdous_et_al:LIPIcs.SEA.2024.12} show empirically that the algorithm by Paz and Schwartzman can compete quality-wise with offline $\frac{1}{2}$-approximation algorithms like GPA, while requiring less memory and time.
Ferdous~et~al.~\cite{ferdous2023streaming} present two semi-streaming algorithms for the related weighted~$k$-disjoint matching problem, building upon the algorithms of  Paz and Schwartzman, and Huang and Sellier~\cite{huang_et_al:LIPIcs.APPROX/RANDOM.2021.14,10.1145/3274668} for streaming~$b$-matching.

\subparagraph{Hypergraph Matching.}
Hazan et~al.~\cite{approxresult} prove that for the maximum~$k$-set packing problem %
there is no approximation within a factor of~$\Omega(k/\log k)$ unless $\mathcal{P}=\mathcal{NP}$. This directly translates to $d$-uniform cardinality hypergraph matching, where every edge has size $d$, with $k=d$ and the number of edges selected is maximized.
Dufosse~et~al.~\cite{dufosse2019effective} investigate reduction rules and a scaling
argument for finding large matchings in~$d$-partite,~$d$-uniform hypergraphs.
There are several approximation results and local search approaches, most notably by Hurkens and Schrijver~\cite{hurkens1989size} and Cygan~\cite{cygan2013improved}  with an approximation guarantee of~$\mathcal{O}(\left(\frac{d+1+\varepsilon}{3}\right))$.
Hanguir and Stein~\cite{hanguir2021distributed} propose three distributed algorithms to compute matchings in hypergraphs, trading off between quality guarantee and number of rounds needed to compute a solution.

The Greedy algorithm for maximum weight hypergraph matching, which considers hyperedges for matching in non-increasing order of weights,  is $1/d$-approximate, where $d$ is the maximum size of a hyperedge~\cite{Besser+:Greedy,Korte+:Greedy}. 
For the weighted~$k$-set packing problem Berman~\cite{10.1007/3-540-44985-X_19} introduces a local search technique.
Improving on these results, Neuwohner~\cite{neuwohner2023passing} presents a way to guarantee an approximation threshold of~$\frac{k}{2}$. We are not aware of any practical implementations of these techniques.
For the more general weighted hypergraph~$b$-matching problem, 
Großmann~et~al.~\cite{grossmann2024engineering} present effective data-reduction rules and local search methods.
In the online setting, when hyperedges arrive in an adversarial order, and one must immediately decide to include the incoming hyperedge or not in the matching, Trobst and Udwani~\cite{trobst2024almost} show that no (randomized) algorithm can have a competitive ratio better than~$\frac{2+o(1)}{d}$.

We are unaware of any studies or implementations for streaming hypergraph matching.

\section{Stack-based Algorithm}

\label{sec:quality}
 We now present our first algorithm to tackle the hypergraph matching problem in the semi-streaming setting. 
Our algorithm uses dual variables to evaluate whether a streaming hyperedge
is good enough to be retained in a stack. %
Once the stream has been ingested, the stack is evaluated from top to bottom to determine a matching, achieving an approximation guarantee of $\frac{1}{d(1+\varepsilon)}$.
We also discuss a more permissive and lenient update function that allows more hyperedges into the stack.
Finally, we discuss the space complexity of our algorithms.

Algorithm~\ref{alg:simple:streaming} shows our framework for computing a hypergraph matching in a streaming setting.
This algorithm is an extension of an algorithm by Paz and Schwartzman~\cite{10.1145/3274668} proposed for graphs.
The algorithm starts with an empty stack and keeps a dual variable~$\phi_v$ for each vertex~$v$ of the hypergraph.
Throughout the algorithm the stack contains candidate hyperedges for inclusion in a matching.

 For each streamed hyperedge~$e$, the algorithm checks if the weight of the dual variables of~$e$'s vertices and thereby the solution can be improved by adding $e$ to the stack. More precisely,  with $\Phi_e := \sum_{v\in e}\phi_v$, the algorithm checks if $\weight(e)\geq \Phi_e(1+\varepsilon)$.
 The parameter~$\varepsilon\in\mathbb{R}_{\geq 0}$ is used to trade quality for memory. A~smaller $\varepsilon$ yields a better approximation guarantee, while a 
 higher~$\varepsilon$ yields a smaller memory consumption. 
 With $\varepsilon=0$, the algorithm no longer provides memory guarantees.
 If $\weight(e)\geq \Phi_e(1+\varepsilon)$, then $e$ is added to the stack, and the dual variables of the vertices of $e$ are then updated using an update function. 
 The update functions that we consider take $\mathcal{O}(1)$ time per vertex.
 
 After all hyperedges have been streamed, a matching containing only the hyperedges stored in the stack is computed.
To do that, our algorithm takes the hyperedges in reverse order from the stack and adds non-overlapping hyperedges to the matching. 
Note that hyperedges processed earlier that have conflicting heavier (later) hyperedges will be ignored. This is crucial to prove the performance guarantee later. 
 The amount of total work needed for scanning one hyperedge~$e$ is $\mathcal{O}(\left|e\right|)$, because we need to sum up the dual variables of $e$'s vertices and update them.

The update functions  %
determine whether an approximation guarantee can be provided and directly impact the results.
Upon processing a hyperedge $e$, the update function applied to the dual variable of a vertex $v\in e$ can use
the prior value~$\phi_v$ and sum~$\Phi_e:=\sum_{v\in e}\phi_v$ as well as $e$'s size and weight.
We define the following update function for proving the approximation guarantee 
{\begin{align}
  \phi_v^{\rm new} :=  \delta_{\mathrm{g}}(e,\phi_v,\Phi_e,\weight(e))& := \phi_v +\overbrace{ \weight(e)-\Phi_e}^{w'_e}. \label{phig}
\end{align}}%

The~$\delta_{\mathrm{g}}$ function is exactly the function used by Paz and Schwartzman~\cite{10.1145/3274668}. 
For a hyperedge $e$ added to the stack, the $\phi(v)$ value of its endpoints is increased by the potential gain in matching weight, $w'_(e)$. This follows since  $\phi(v)$ stores the gain in weight of all earlier edges incident on $v$ that have been added to the stack thus far. 

We introduce a different update function later.

At the core of this method are the variables~$\phi_v$ for each~$v$ and the update mechanism.
These core components and local ratio techniques are used to prove the approximation guarantee in the original work~\cite{10.1145/3274668}.
We follow the structure of the proof derived from a primal-dual analysis~\cite{ghaffari2017simplified}.

\begin{algorithm}[t]
{%
  \begin{algorithmic}[1]
    \PROCEDURE{StackStreaming}{$H=(V,E,\weight)$}
    \STATE~$S\gets emptystack$
    \STATE $\forall v\in V: \phi_v=0$
    \FOR{$e\in E$ in any (even adverse) order}
    \STATE $\Phi_e \gets \sum_{v\in e}\phi_v$
    \IF{$\weight(e)<\Phi_e(1+\varepsilon)$}\label{alg:simple:streaming:scan}
    \STATE next
    \ENDIF
    \STATE $S.push(e)$
    \FOR{$v\in e$}
      \STATE $\phi_v \gets \delta(e,\phi_v,\Phi_e,\weight(e))$ \COMMENT{update}\label{alg:simple:streaming:update}
    \ENDFOR
    \ENDFOR
    \STATE $M \gets \emptyset$
    \WHILE{$S\neq \emptyset$}
    \STATE $e \gets S.pop()$
    \IF{$\forall f\in M: f\cap e=\emptyset$  }
    \STATE$M\gets M \cup \{e\}$
    \ENDIF
    \ENDWHILE
    \ENDPROCEDURE
  \end{algorithmic}}
  \caption{Simple Streaming Algorithm.}
  \label{alg:simple:streaming}
\end{algorithm}
\subsection{Approximation Guarantee}\label{subsec:dual:problem}
We make use of the primal-dual framework for showing the approximation guarantee.
For general linear programming concepts, we refer the reader to books \cite{bazaraa2011linear,wright1997primal}, and for an excellent overview of primal-dual approximation method to the book~\cite{williamson2011design}.

We show that using~$\delta_{\mathrm{g}}$ update function~\eqref{phig} in our algorithm leads to a~$\frac{1}{d(1+\varepsilon)}$-approximation, when the stack is~unwound and hyperedges are selected in that order which they are on the stack.
When streaming the hyperedges in descending order of weights, we can prove that  the Greedy algorithm achieves a $1/d$ approximation factor.

We now proceed with a primal-dual analysis of an ILP formulation of the hypergraph matching problem.
In our formulation, there is a binary decision variable associated with each hyperedge to designate if that hyperedge is selected to be in the matching.
The objective function is to maximize the sum of the weights of the selected hyperedges.
 The constraint is to select at most one hyperedge containing a given vertex.
The linear programming relaxation, shown in Figure~\ref{fig:lp},  is obtained by dropping the binary constraint on the hyperedge variables.
\begin{figure}
\caption{Primal and Dual LPs for Hypergraph Matching.}
    \label{fig:enter-label}
    \centering
    \begin{subfigure}[t]{0.48\linewidth}
\subcaption{LP}
  {\small \begin{align*}
  \mathrm{maximize}\quad & \sum_{e\in E} \weight(e)x_e\\
  \mathrm{subject\, to }\quad & \\
  \forall v \in V\colon& \sum_{e\ni v}x_e \leq 1\\
  \forall e\in E\colon& x_e\geq 0.
\end{align*}}%
\label{fig:lp}
    \end{subfigure}
    \begin{subfigure}[t]{0.48\linewidth}
    \subcaption{Dual LP}
\label{fig:dual}
    {\small
        \begin{align}
  \mathrm{minimize}\quad & \sum_{v\in V} \phi_v\nonumber\\
  \mathrm{subject\, to }\quad & \nonumber\\
  \forall e \in E\colon& \sum_{v\in e}\phi_v \geq \weight(e)\label{eq:term:dual}\\
  \forall v\in V\colon& \phi_v\geq 0.\nonumber
\end{align}
}
    \end{subfigure}
\end{figure}

The objective value of  the relaxation is naturally greater or equal than the integer version of the linear program. 
The dual problem of the relaxed hypergraph matching is given in Figure~\ref{fig:dual}.
Following the weak duality theorem for linear programs, 
we know that any feasible solution of the dual has an objective value greater or equal to the objective value of any feasible primal solution.
 Furthermore, the optimal value of the linear program~$\mathrm{OPT}(LP)$ is equal for both problems~(strong~duality).
The first step for the proof is to check that the variables $(1+\varepsilon)\sum_{v\in e}\phi_v$ 
(from  Algorithm~\ref{alg:simple:streaming})
constitute a valid dual solution.

\subparagraph{Observation.} Function~$\delta_{\mathrm{g}}$ and Algorithm~\ref{alg:simple:streaming} generate valid dual solutions for all~$\varepsilon\geq 0$.
\begin{proof}

 For each hyperedge $e$ not on the stack, there was enough weight in the~$\phi_v$ values of its vertices, when~$e$ was scanned in Line~\ref{alg:simple:streaming:scan}.
 In the update for every added hyperedge to the stack, all vertices $\phi_v$ values are increased by $w(e)-\Phi_e$ such that clearly the sum of vertex $\phi_v$ values is higher than the weight of the hyperedge just added. 
 Therefore, for any hyperedge it holds~$\sum_{v\in e}\phi_v\geq \weight(e)$, satisfying the dual equation~\eqref{eq:term:dual}. 
\end{proof}

Such a valid dual solution has a greater objective value then the optimum solution of the relaxed dual, and the LP duality theorem gives an upper bound for every matching, including the optimal one ~$M^\ast$ by 
$\weight(M^\ast)\leq \mathrm{OPT}(LP)\leq (1+\varepsilon)\sum_v\phi_v.$

Now, we connect the changes to the dual variables with the hyperedges that have already been processed. 
 Define~
 \begin{align}
 \label{def:delta}
\Delta\phi^{e} =  \begin{cases}\sum_{v\in e}(\delta_{\mathrm{g}}(e,\phi_v,\Phi_e,\weight(e))-\phi_v)= \left|e\right|(\weight(e)-\Phi_e)& \text{if }e\in S\\
0& \text{else}\end{cases}\end{align}
as the change to the dual~$\sum_v\phi_v$ by inspecting~$e$.
 We  give a bound for the change of the dual variable w.r.t.~the preceding hyperedges in Lemma~\ref{lemma:prev:insp}.
 
\begin{lemma}
  \label{lemma:prev:insp}
  For a hyperedge~$e$, let~$\weight'_{e}:= \weight(e)-\Phi_{e}$.
  For each hyperedge~$e\in E$ added to stack~$S$, if we denote its preceding neighboring hyperedges (including itself) by~$\mathcal{P}(e):=\{c\mid c\cap e\neq \emptyset,c \text{ added before }\, e\}\cup\{e\}$, then 
 $\weight(e)\geq \sum_{e'\in \mathcal{P}(e)}\dfrac{1}{d}\Delta\phi^{e'} = \sum_{e'\in \mathcal{P}(e)} \weight'_{e'} .$
\end{lemma}
\begin{proof}
From the definition~\eqref{def:delta}, we have $\Delta\phi^e = \left| e\right| \weight'_{e}$ for~$\delta_{\mathrm{g}}$, because of line~\ref{alg:simple:streaming:update} of Algorithm~\ref{alg:simple:streaming}. 
$\Phi_e:= \sum_{v\in e}\phi_v$ is defined as the previous value of the dual variables  before inspecting~$e$. Each of these dual values $\phi_v$ consists of the sum $\sum_{c\in \mathcal{P}(e)\text{ s.t. }v\in c} \weight'_{c}$ for all preceding hyperedges. This leads to 
$\Phi_e= \sum_{v\in e}\phi_v  \geq \sum_{c\in \mathcal{P}(e)\setminus\{e\}}\frac{1}{\left|e\right|}\Delta\phi^{e'}\geq \sum_{c\in \mathcal{P}(e)\setminus\{e\}}\frac{1}{d}\Delta\phi^{c}$. 
So we can conclude \\ 
$\weight(e) = \weight'_e+ \Phi_e \geq \frac{1}{d}\Delta\phi^{e}+ \sum_{e'\in \mathcal{P}(e)\setminus\{e\}}\frac{1}{d}\Delta\phi^{e'}
$.
\end{proof}
The previous bound relates the weight of a hyperedge to the change in dual variables by its predecessors. In conclusion,
 we show that  our algorithm returns a~$\frac{1}{d(1+\varepsilon)}$-approximation.
\begin{lemma}
   Algorithm~\ref{alg:simple:streaming} with $\delta_{\mathrm{g}}$ function guarantees a~$\frac{1}{d(1+\varepsilon)}$-approximation. \label{lemma:dapprox}
\end{lemma}
\begin{proof}
  We show a lower bound on the weight of any matching $M$ constructed by the algorithm.
  For any hyperedge~$e$ not in the stack,  ~$\Delta\phi^e=0$, as when a hyperedge is not pushed into the stack, no dual variables are updated.
 Furthermore,
 any hyperedge in the stack that is not included in the matching must be a previously added neighbor of a matching hyperedge as defined in Lemma~\ref{lemma:prev:insp}.
  Therefore, 
  Lemma~\ref{lemma:prev:insp} applies, 
  and the weight can be lower-bounded.
   The sum of changes $\sum_{e}\Delta\phi^e$ to $\phi$ is equal to the sum of dual variables $\sum_v \phi_v$ at the end.
  We have
  
  \vspace{0.25cm}\noindent
  {$ \weight(M)  = \sum_{e\in M}\weight(e) \overset{\text{L.~\ref{lemma:prev:insp}}}{\geq} \sum_e\dfrac{1}{d}\Delta\phi^{e} 
                     \geq \dfrac{1}{d}\sum_e\Delta\phi^{e}=\dfrac{1}{d}\sum_v\phi_v \overset{\text{LP Duality}}{\geq} \dfrac{1}{d(1+ \varepsilon)} \weight(M^\ast)$}.
\end{proof}

\subsection{Improving Solution Quality} 
\label{subsec:ils}
  Now we look into optimizing the solution quality. 
  The design space for optimizations is vast, therefore we focus on simple yet effective techniques.
  We propose a second update function  allowing more hyperedges than $\delta_{\mathrm{g}}$ into the stack, with the same approximation guarantee.

\subparagraph{Lenient Update Function.}
The $\delta_{\mathrm{g}}$ in Equation~\ref{phig}  builds a dual solution much larger than needed. For every successfully added hyperedge, the difference between the current dual solution and the weight of the hyperedge is added to every vertex of the hyperedge.
 We address this by combining it with a scaling argument. The resulting function is
\begin{align}
  \delta_{\mathrm{lenient}}(e,\phi_v,\Phi_e,\weight(e)):= \phi_v+(\weight(e)-\Phi_e)/\left|e\right|. 
\label{phil}
\end{align}
This function produces a valid dual solution, and Lemma~\ref{lemma:prev:insp} also holds.
For every previously added neighboring hyperedge $e'$ the change $\Delta\phi^{e}=\weight'_{e'}$ was distributed over all vertices of the hyperedge, so $\Phi_e\geq\sum_{e'\text{ added before}}\frac{1}{d}\Delta\phi^{e'}$. Lemma~\ref{lemma:dapprox} follows and gives us the desired approximation factor of $\frac{1}{d(1+\varepsilon)}$.
\vspace{-0.25cm}
\subsection{Space Complexity Analysis} \label{subsec:space}
The space complexity of $\delta_{\mathrm{g}}$ and $\delta_{\mathrm{lenient}}$ can be deduced by simple counting arguments.
Let $W$ be the maximum normalized weight of a hyperedge in the hypergraph, i.e.,~$W:=\frac{\max_{e\in E}{\weight(e)}}{\min_{e\in E}{\weight(e)}}$, and let $W$ be $\mathcal{O}(\mathrm{poly}(n))$. We discuss both update functions separately.
\subparagraph{Guarantee Function.}
On every vertex we can observe up to $1+\log_{1+\varepsilon}(W)$ incrementing events, because every change in the dual variables has to be bigger by a factor of $(1+\varepsilon)$.
 This causes the stack to contain  $\mathcal{O}(\sum_v(1+\log{W}/\varepsilon))$ vertices in its edges. %
 Each vertex in an edge requires $\log{n}$ bits.
Therefore, the overall space complexity is $\mathcal{O}((1/\varepsilon) n\log^2{n})$ bits since $W$ is $\mathcal{O}(poly(n))$.
 
\subparagraph{Lenient Function.}
This function updates  $\phi_v$ for every vertex in an edge $e$ added to the stack by $\frac{\weight(e)-\Phi_v}{d}$,  resulting in $1+d\cdot\log_{1+\varepsilon}(W)$ possible increases to reach the total sum of $W$. 
Following the same argument, the stack's space complexity in bits is~$\mathcal{O}( (1/\varepsilon) nd\log^2{n})$. Note that with this update function the algorithm is semi-streaming only if $d$ is $\mathcal{O}(\mathrm{polylog}(n))$.

The time complexity per edge is in $\mathcal{O}(d)$ and $\Omega(md)$ overall in the scanning phase, as for every edge its vertices need to be scanned and optionally updated once. The unwinding of the stack takes $d$ checks per hyperedge, when reusing the memory of $\phi_v$ from the previous step. Overall, since the stack size is naturally bounded by $m$, the complexity of this algorithm is $\mathcal{O}(md)$.
In Section~\ref{sec:rq1} we show the difference in the stack size in experiments under several orderings of the input.
\begin{algorithm}[t]
  \caption{Swapping streaming algorithm.}
  \label{alg:simple:stackless}
   \algsetup{linenosize=\small}
  {%
  \begin{algorithmic}[1]
    \PROCEDURE{SwapSet}{($H=(V,E,\weight)$, $\alpha$)}
    \STATE $\forall v \in V\colon\bestEdgeSymbol_v = \bot$\COMMENT{{ Initialize best hyperedge to empty}}
    \STATE $\weight(\bot):=0$
    \FOR{$e\in E$ in arbitrary order}
    \STATE $C\gets \bigcup_{v\in e} \bestEdgeSymbol_v$
    \STATE $\Phi_e \gets \weight(C)$\label{line:upper:sum}\COMMENT{{ weight of hyperedges to be removed.}}
    \IF{$\weight(e)\geq(1+\alpha)\cdot\Phi_e$ }
    \FOR{$v\in e$}
    \IF{$\bestEdgeSymbol_v\neq \bot$}
    \FOR{$w\in \bestEdgeSymbol_v$}
    \STATE $\bestEdgeSymbol_w\gets \bot$  \COMMENT{{ Unmatch vertices in $\bestEdgeSymbol$.}}
    \ENDFOR 
    \ENDIF
      \STATE $\bestEdgeSymbol_v \gets e$
    \ENDFOR
    \ENDIF
    \ENDFOR
    \STATE $M\gets \bigcup_{v\in V}\bestEdgeSymbol_v$
    \RETURN $M$
    \ENDPROCEDURE
  \end{algorithmic}}
\end{algorithm}
\section{Greedy Swapping 
Algorithm}\label{sec:stackless}

We now propose a second streaming algorithm that computes, stores, and updates a matching in the hypergraph as the edges stream. It is conceptually similar to a streaming matching algorithm for graphs~\cite{mcgregor2005finding}.
 It requires a constant amount of memory per vertex, has an approximation factor that depends on the maximum size of a hyperedge $d$ and a parameter $\alpha$ 
 and obtains high-quality matchings in practice.

 The proposed approach is described in Algorithm~\ref{alg:simple:stackless}.
  We store for every vertex $v$ a reference to the current matching hyperedge containing $v$; the~$\bot$ sign symbolizes that no matching hyperedge contains $v$.
  For simplifying the presentation, we 
 define~$\weight(\bot)=0$.
  When we inspect a hyperedge $e$, we compute the sum of the weights of its adjacent hyperedges that are currently in the matching.
 If the weight of the incoming hyperedge is larger than $(1+\alpha)$ times the previous conflicting hyperedges, we first remove the previous hyperedges from all their vertices.
 Afterwards we can set the reference to the new incoming hyperedge. %
The overall space consumption of this algorithm is~$\mathcal{O}(n)$.
There are~$n$ references involved, and each hyperedge of size~$d$ holds the~$d$ vertices referencing it. 
Removing a hyperedge is linear~in~$d$, because up to $d$ vertices in $\mathcal{B}_v$ must be set to $\bot$; 
it follows that the total work~is~$\mathcal{O}(d \cdot |E|)$.

We show its approximation guarantee of $\frac{1}{(1+\alpha)(\frac{d-1}{\alpha}+d)}$, which is optimal for $\alpha=\sqrt{(d-1)/d}$, in the Appendix~\ref{sec:stackless:quality}.
In our experimental evaluation, we look at various values of $\alpha\in \{0,0.1,1\}$ and \hbox{$\alpha=\sqrt{(d-1)/d}$}. %
Note that for $\alpha=0$ the algorithm has no guarantee.
\vspace{-0.25cm}
\section{Experimental Evaluation}\label{sec:experiments}
We now evaluate our algorithms with respect to solution quality, running time, and memory usage. Specifically, we address the following research questions:

\begin{itemize}
    \item \textbf{RQ1:} How does the ordering of the hyperedges affect our metrics
(memory, running time, and quality)?
\item \textbf{RQ2:} How do other instance properties affect our algorithms in their memory needs? 
\item \textbf{RQ3:} How do our algorithms compare with offline greedy or naive streaming approaches?
\end{itemize}

\subsection{Setup and Data Set}
We implemented our approaches in \cpp using  \textsf{g}\texttt{++-}\textsf{14.2} with full optimization turned on (\texttt{-O3} flag).
We tested on two identical machines, equipped with 128 GB of main memory and a Xeon w5-3435X processor running at 3.10 GHz having a L3~cache of 45 MB each.
We repeat each experiment three times, and the results are compared only if experiments are run on the same machine and the same compute job.
For memory consumption, we used the jemalloc malloc implementation~\cite{jemalloc}.
The time needed for loading the hypergraph is not measured.
 We scheduled eight experiments (RQ1) and ten experiments (RQ2) to run in parallel, and the order of the experiments was randomized.
 In order to compare the results we use performance profiles as suggested by Dolan and Moré~\cite{DBLP:journals/mp/DolanM02}. 
  We plot the fraction of instances that could be solved within a factor $\tau<1$ of the best result per instance.
 In the plot (Figure~\ref{fig:res:rq3:size}),  the algorithm towards the top left corner is the best performer.

 Our benchmark includes general hypergraphs that are primarily used for partitioning and social link hypergraphs generated from question-answering websites.
In social link hypergraphs, a hypergraph matching can be used to summarize the overall websites, as it represents a subset of disjoint pages from different categories or users. %
In hypergraph partitioning, matchings can be used to contract the hypergraph in a multilevel scheme.
In a social hypergraph, each page (e.g., a post on StackOverflow or an article in Wikipedia) is considered as a hyperedge. %
The (\textsc{Threads}) graphs model participating users as vertices, whereas the (\textsc{Tags}) ones model the tags of the posts as vertices.
We use three stackexchange networks by Benson~et~al.~\cite{Benson-2018-simplicial}, and set the number of views as weights from~\cite{datastackexchange}.
For the StackOverflow instance, the weights are set by querying the public dataset from BigQuery~\cite{googlebigquerypublicdata}.
 We created an additional instance  from this source.
 We also generated a new hypergraph from the English Wikipedia dump, where the categories represent the vertices and the articles represent the hyperedges. We selected a category as a vertex if it has at least 25 mentions, which resulted in 293K vertices for 8M articles. The access frequency of each article in December 2024 is set as weight. We answer \textbf{RQ1} and \textbf{RQ3} with this data set.

We use the hypergraph data set $L_{HG}$  
collected earlier~\cite{DBLP:journals/corr/abs-2303-17679} for hypergraph partitioning to address \textbf{RQ2} and \textbf{RQ3}. 
The set consists of~94~instances,  spanning a wide range of applications from DAC routability-driven placement~\cite{dacset}, general matrices~\cite{sparsecollection} to SAT solving~\cite{satbenchmark}.
As weights we use $\max\left|e\right|-\left|e\right|$. This function maximizes the cardinality (number of edges matched).
 These instances have up to~$1.4\times 10^8$ hyperedges/-vertices and a maximum hyperedge size of %
 $2.3\times 10^6$ vertices.
 More statistics are in Table~\ref{tab:stats} and \ref{tab:stats:social}  in the Appendix.

For the stack-based approach of Section~\ref{sec:quality}, we implemented the \textsf{Guarantee} and \textsf{Lenient} update functions.
For simplicity, we refer to the algorithm with the stricter \textsf{Guarantee} function as \textsf{Stack} and the one with the lenient function as \textsf{StackLenient}.
Both approaches can be configured by their~$\varepsilon$ parameter.
  The algorithms from Section~\ref{sec:stackless} are named \textsf{SwapSet}. The \textsf{SwapSet} algorithm has a parameter $\alpha\geq 0$.
 For comparison, we use a non-streaming \textsf{Greedy} algorithm that sorts the hyperedges based on weight in descending order and greedily adds them to a matching.
We also implemented a \textsf{Naive} streaming algorithm that maintains a maximal hypergraph matching (i.e, it includes a hyperedge if it is feasible w.r.t the current matching) in the order that they are streamed.
\subsection{Impact of Streaming Order}
\label{sec:rq1}
We now investigate the impact of different streaming orders on our algorithms (\textbf{RQ1}). To this end, we stream the hyperedges of the social-link hypergraphs in three ways: by ascending weight, descending weight, and the original input order. The original input order is the order given by either the original files~\cite{Benson-2018-simplicial} or, in the case of the wikipedia instance, the order that the articles appear in the dump. 
\subparagraph{Memory.} Figure~\ref{fig:res:rq1:mem} shows the geometric mean memory consumption for all approaches grouped by the ordering.
Interestingly, for lower values of $\varepsilon$, the memory consumption is significantly higher for the ascending order, up to 2.97 times over the descending order (\textsf{StackLenient} $\varepsilon=0$). 
For the swapping based algorithms (\textsf{SwapSet}) the differences between the orderings are only minimal, nearly reaching the memory consumption of the \textsf{Naive} algorithm.
The \textsf{Greedy} algorithm utilizes the same amount of memory for all orderings, 13.43 times more memory than the stack approaches in the geometric mean.
The ordering heavily affects our stack-based algorithms, but they still require less memory than the \textsf{Greedy} approach.
The \textsf{SwapSet} algorithms are not significantly impacted by the order.
Both results are in line with the theoretical results derived in Section~\ref{subsec:space} and \ref{sec:stackless}.
\begin{figure}
    \centering
    \resizebox{\linewidth}{!}{\input{figures/experiments/esa-experiments/memory_catbar_0_all}}
    
    \caption{Geometric mean of the memory consumption on the social-link hypergraphs. Note the $\log{}$-scale on the $y$-axis.}
    \label{fig:res:rq1:mem}
\end{figure}
\subparagraph{Running time.}
In Figure~\ref{fig:res:rq1:time} the geometric mean of the running times of the compared algorithms are shown.
The \textsf{Naive} streaming algorithm is the fastest, requiring similar time over every ordering.
The ascending ordering requires more time for both greedy swapping and stack-based approaches since more replacements in the swapping algorithm and placements in the stack happen in that order. 
The stack-based approaches have a running time only 26\% higher than the \textsf{Naive}.
We observe that preordering the hyperedges speeds up the offline \textsf{Greedy} algorithm by a factor of 2.73 in comparison to the original ordering, but it is still slower than the stack-based algorithms.
This is due to our use of \textsf{std::sort} that is partly optimized for ordered data. 
The running time of the \textsf{SwapSet} algorithms is comparable to the running time of \textsf{Greedy} with natural ordering.
The cost for voiding edges in the \textsf{SwapSet} algorithm is higher than simply pushing edges on the stack.
In general, the streaming process of \textsf{Stack} and \textsf{StackLenient} is similar to that of \textsf{Naive}.
The small overhead is rooted in the required pass through the stack to build the  matching. %
 The ordering affects how many edges are added to the stack (see the following Memory section) and so the running time as well, but not as much as in the \textsf{Greedy}~algorithm.
 
\begin{figure}
    \centering
    \resizebox{\linewidth}{!}{\input{figures/experiments/esa-experiments/runtime_catbar_0_all}}
    \caption{Geometric mean of running times on the social-link hypergraphs.}
    \label{fig:res:rq1:time}
\end{figure}

\subparagraph{Quality.} We present the geometric mean of the weights of matchings in Figure~\ref{fig:res:rq1:weight}. %
As expected, the \textsf{Greedy} algorithm is not affected by the ordering since it computes its own order in the beginning. The quality of \textsf{Naive} algorithm is heavily affected by orderings, where the best weight is achieved when the hyperedges are streamed in descending order. On ascending and original order, the quality of the \textsf{Naive} is significantly worse.
Under the descending orderings, the swapping algorithm (\textsf{SwapSet}) produces the same results as the \textsf{Greedy} algorithm, whereas the ascending order produces worse quality results in the \textsf{SwapSet} algorithm. %
The \textsf{Stack} and \textsf{StackLenient} approaches are more robust under different orderings and compute better quality matchings in the adversarial ascending order compared to other algorithms..
\begin{figure}
    \centering
    \resizebox{\linewidth}{!}{\input{figures/experiments/esa-experiments/weight_catbar_0_all}}

    \caption{Geometric mean of weights on the social-link hypergraphs.}
    \label{fig:res:rq1:weight}
\end{figure}
\subsection{Impact of Other Properties}
\label{sec:rq2}
We now test whether some other properties of the instances have an impact on the memory consumption of the algorithm (\textbf{RQ2}).
Namely, we check if the number of vertices and the number of pins (sum of all edge sizes) are correlated with the memory consumption.
We use the data set by Gottesbüren~et~al.~\cite{DBLP:journals/corr/abs-2303-17679} from hypergraph partitioning, which contains many larger-scale instances.
 This data set contains 94 diverse instances, making it more suitable for statistical testing.
  The ordering is the original ordering as given by the files in~\cite{DBLP:journals/corr/abs-2303-17679}.
Instance details can be found in the Appendix in Table~\ref{tab:stats}.
A hyperedge $e$ in these instances has a weight of  $(\max_{e\in E}\left|e\right|)-\left|e\right|$, optimizing the number of hyperedges matched. %
 We set $\varepsilon$ (and also $\alpha$) to $\{0,1\}$ since they represent the two extremes in terms of memory~usage.

\subparagraph{Number of Vertices.}
Figure~\ref{fig:res:rq2:n} shows the memory consumption plotted against the number of vertices in each instance. %
The plot is on a $\log$-$\log$-scale.
Naturally, \textsf{Naive} requires the least amount of memory.
Our \textsf{SwapSet} algorithm follows, as well as the \textsf{Stack} algorithm and the \textsf{StackLenient}.
Finally, the offline \textsf{Greedy} requires even more memory.
This is also backed by the theoretical results of Sections~\ref{subsec:space} and \ref{sec:stackless}.
The higher the number of vertices, the smaller is the difference in magnitudes between the approaches. %
This is due to the additional memory needed for the list of finally matched hyperedges and the overhead in some allocations.
The measured peak allocation may contain some overallocated memory caused by the growing result vectors.
\begin{figure}
    \centering
    \includegraphics[width=0.8\linewidth]{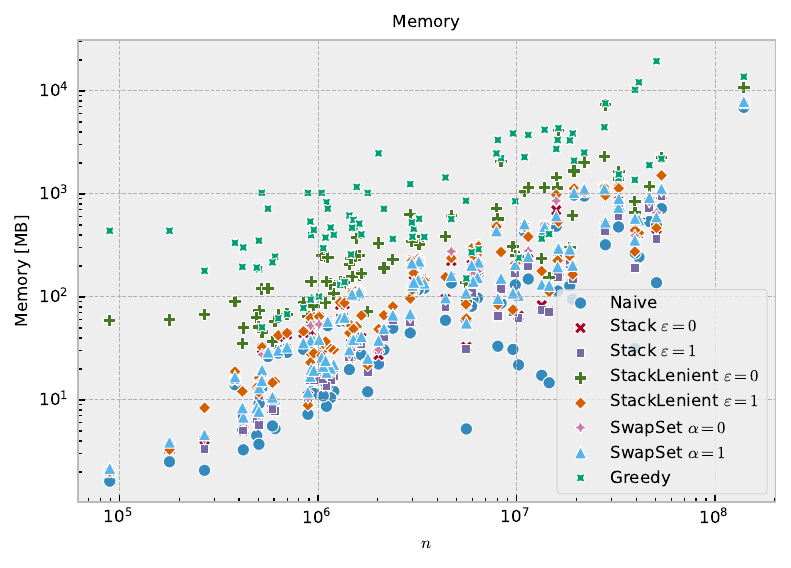}
    \caption{Geometric mean of the memory consumption on hypergraphs from partitioning plotted against the number of vertices.}
    \label{fig:res:rq2:n}
\end{figure}
\subparagraph{Number of Pins.}
In Figure~\ref{fig:res:rq2:pins}, a plot of the geometric mean memory consumption against the number of pins (sum of all edge sizes) for our approaches and the competitors is shown.
Additionally, we added a linear regression line for the \textsf{Greedy} algorithm's memory consumption.
The linear model achieves an $R^2$ score of $0.96$ when (randomly) splitting the data set into a training ($n=75$) and evaluation test set ($n=19$), on the latter.
This shows that the greedy algorithm requires memory linear in the number of pins since it loads the whole hypergraph at the start.
Our proposed algorithms require less memory, especially the  \textsf{SwapSet} algorithms nearly match the memory consumption of the \textsf{Naive} approach.
\begin{figure}
    \centering
    \includegraphics[width=0.8\linewidth]{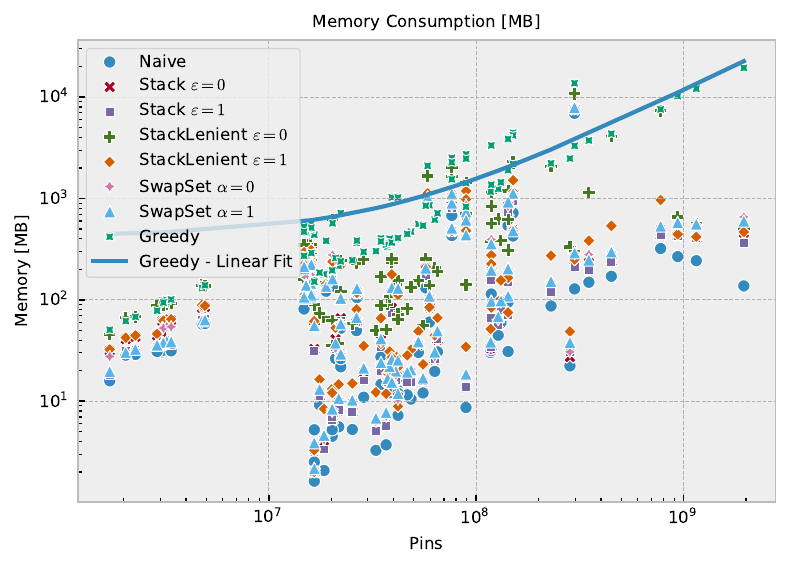}
    \caption{Geometric mean of the memory consumption on the  hypergraphs from partitioning. Linear Regression line for the offline \textsf{Greedy} algorithm.}
    \label{fig:res:rq2:pins}
\end{figure}
\subparagraph{Other Properties.}
We also tested the dependence of the memory consumption on the average hyperedge size and the number of hyperedges, but found no pattern.
 Note that the number of pins is equal to the product of the number of edges and the average hyperedge size.

\subsection{Comparison with Offline Greedy and Naive Streaming}
\label{sec:rq3}
 In this section, we compare the algorithms on the instances stemming from hypergraph partitioning (\textbf{RQ3}). The  Results  on the social set of instances can be found in Section~\ref{sec:rq1}.
Figure~\ref{fig:res:rq3:size} shows a performance profile for the cardinality/size of the matchings for our algorithms as well as the \textsf{Naive} streaming and offline \textsf{Greedy} algorithm.
The \textsf{StackLenient} variant with~$\varepsilon=0$ computes the biggest matchings and is the best performing algorithm, followed by \textsf{SwapSet}~$\alpha=0$.
The \textsf{Greedy} algorithm has similar results to \textsf{StackLenient}~$\varepsilon=1$.
The \textsf{Stack}~$\varepsilon=1$ and \textsf{SwapSet}~$\alpha=1$ are our worst-performing algorithms, having results very close to the \textsf{Naive} streaming algorithm. 
Their higher parameters ($\alpha=\varepsilon=1$) cause the algorithm to converge to the naive algorithm.

\begin{figure}
    \centering
    \includegraphics[width=0.8\linewidth]{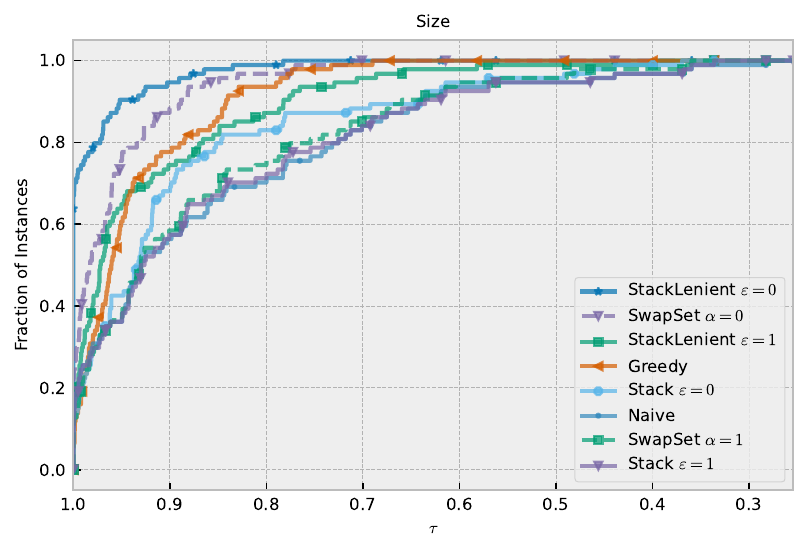}
    \caption{Performance Profile for the size of the matching in  hypergraphs used in partitioning.}
    \label{fig:res:rq3:size}
\end{figure}

\section{Conclusion}\label{sec:conclusion}
We have proposed two (semi-)streaming algorithms for hypergraph matching.
The first, inspired by Paz-Schwartzman~\cite{10.1145/3274668}, uses a stack and adds hyperedges to it, based on dual variables it keeps track of and updates according to an update function.
The approximation guarantee for this algorithm is $\frac{1}{d(1+\varepsilon)}$, and its running time per hyperedge is linear in hyperedge size.
We have proposed two other update functions to be used in this algorithm.
The proposed update functions result in a space complexity of $\mathcal{O}((1/\varepsilon)n\log^2{n})$ and $\mathcal{O}((1/\varepsilon)nd\log^2{n})$ bits.
The second proposed algorithm works by greedily swapping out hyperedges and maintaining only one solution, requiring only $\mathcal{O}(n)$ memory.
Inspired by McGregor's $\frac{1}{3+2\sqrt{2}}$-approximation guarantee~\cite{mcgregor2005finding}, we have shown that if every swap increases the weight by at least $(1+\alpha)$, the algorithm has an approximation guarantee of $\frac{1}{(1+\alpha)\left(\frac{d-1}{\alpha}+d\right)}$. The best choice is~$\alpha=\sqrt{(d-1)/d}$. 

 In extensive experiments, we have shown the competitiveness of the proposed algorithms in comparison to the standard non-streaming \textsf{Greedy} and a \textsf{Naive} streaming approach with respect to running time, memory consumption, and quality.
We showed that the order of the hyperedges in the stream directly impacts the solution quality and that our stack algorithms handle worst-case orderings (like ascending weights for the \textsf{Naive} algorithm) even better than the offline \textsf{Greedy} algorithm.
The running times of our \textsf{Stack} algorithm are only 26\% higher than the \textsf{Naive} algorithm.
Furthermore, we validated that the memory consumption for both families of algorithms is not linear in the number of pins, as it is for the \textsf{Greedy} algorithm; for the \textsf{Stack} approaches on the social link hypergraphs it is 13 times lower than \textsf{Greedy}.
Lastly, we report that our algorithm obtains considerably better cardinality matchings in general hypergraphs from partitioning tasks.

Avenues of future work include improving the solution quality and extending to problems with relaxed capacity constraints.
  We aim to develop a streaming algorithm that can efficiently handle instances with capacity $\scriptsize{b(v)>1}$ at each vertex, while maintaining a reasonable approximation ratio and computational~overhead.
\bibliography{compact}

\begin{thebibliography}{10}

\bibitem{bazaraa2011linear}
Mokhtar~S Bazaraa, John~J Jarvis, and Hanif~D Sherali.
\newblock {\em {Linear} {Programming} and {Network} {Flows}}.
\newblock John Wiley \& Sons, 2011.

\bibitem{satbenchmark}
Anton Belov, Daniel Diepold, Marijn Heule, and Matti J\"arvisalo.
\newblock The {SAT} competition 2014.
\newblock \url{http://www.satcompetition.org/2014/index.shtml}, 2014.

\bibitem{Benson-2018-simplicial}
Austin~R. Benson, Rediet Abebe, Michael~T. Schaub, Ali Jadbabaie, and Jon~M.
  Kleinberg.
\newblock Simplicial closure and higher-order link prediction.
\newblock {\em Proc. Natl. Acad. Sci. {USA}}, 115(48):E11221--E11230, 2018.
\newblock \href {https://doi.org/10.1073/pnas.1800683115}
  {\path{doi:10.1073/pnas.1800683115}}.

\bibitem{10.1007/3-540-44985-X_19}
Piotr Berman.
\newblock A \emph{d}/2 approximation for maximum weight independent set in
  \emph{d}-claw free graphs.
\newblock In Magn{\'{u}}s~M. Halld{\'{o}}rsson, editor, {\em Algorithm Theory -
  {SWAT} 2000, 7th Scandinavian Workshop on Algorithm Theory, Bergen, Norway,
  July 5-7, 2000, Proceedings}, volume 1851 of {\em Lecture Notes in Computer
  Science}, pages 214--219, Berlin, Heidelberg, 2000. Springer.
\newblock \href {https://doi.org/10.1007/3-540-44985-X\_19}
  {\path{doi:10.1007/3-540-44985-X\_19}}.

\bibitem{Besser+:Greedy}
B.~Besser and M.~Poloczek.
\newblock Greedy matching: {Guarantees} and limitations.
\newblock {\em Algorithmica}, 77:201--234, 2017.

\bibitem{birn2013efficient}
Marcel Birn, Vitaly Osipov, Peter Sanders, Christian Schulz, and Nodari
  Sitchinava.
\newblock Efficient parallel and external matching.
\newblock In Felix Wolf, Bernd Mohr, and Dieter an~Mey, editors, {\em Euro-Par
  2013 Parallel Processing - 19th International Conference, Aachen, Germany,
  August 26-30, 2013. Proceedings}, volume 8097 of {\em Lecture Notes in
  Computer Science}, pages 659--670. Springer, Springer, 2013.
\newblock \href {https://doi.org/10.1007/978-3-642-40047-6\_66}
  {\path{doi:10.1007/978-3-642-40047-6\_66}}.

\bibitem{cygan2013improved}
Marek Cygan.
\newblock Improved approximation for 3-dimensional matching via bounded
  pathwidth local search.
\newblock In {\em 54th Annual {IEEE} Symposium on Foundations of Computer
  Science, {FOCS} 2013, 26-29 October, 2013, Berkeley, CA, {USA}}, pages
  509--518. IEEE, {IEEE} Computer Society, 2013.
\newblock \href {https://doi.org/10.1109/FOCS.2013.61}
  {\path{doi:10.1109/FOCS.2013.61}}.

\bibitem{sparsecollection}
Timothy~A. Davis and Yifan Hu.
\newblock The university of florida sparse matrix collection.
\newblock {\em {ACM} Trans. Math. Softw.}, 38(1):1:1--1:25, Dec. 2011.
\newblock \href {https://doi.org/10.1145/2049662.2049663}
  {\path{doi:10.1145/2049662.2049663}}.

\bibitem{DBLP:journals/mp/DolanM02}
Elizabeth~D. Dolan and Jorge~J. Mor{\'{e}}.
\newblock Benchmarking optimization software with performance profiles.
\newblock {\em Math. Program.}, 91(2):201--213, 2002.
\newblock \href {https://doi.org/10.1007/s101070100263}
  {\path{doi:10.1007/s101070100263}}.

\bibitem{drake2003simple}
Doratha~E. Drake and Stefan Hougardy.
\newblock A simple approximation algorithm for the weighted matching problem.
\newblock {\em Inf. Process. Lett.}, 85(4):211--213, 2003.
\newblock \href {https://doi.org/10.1016/S0020-0190(02)00393-9}
  {\path{doi:10.1016/S0020-0190(02)00393-9}}.

\bibitem{dufosse2019effective}
Fanny Dufoss{\'{e}}, Kamer Kaya, Ioannis Panagiotas, and Bora U{\c{c}}ar.
\newblock Effective heuristics for matchings in hypergraphs.
\newblock In Ilias~S. Kotsireas, Panos~M. Pardalos, Konstantinos~E.
  Parsopoulos, Dimitris Souravlias, and Arsenis Tsokas, editors, {\em Analysis
  of Experimental Algorithms - Special Event, SEA{\({^2}\)} 2019, Kalamata,
  Greece, June 24-29, 2019, Revised Selected Papers}, volume 11544 of {\em
  Lecture Notes in Computer Science}, pages 248--264. Springer, Springer, 2019.
\newblock \href {https://doi.org/10.1007/978-3-030-34029-2\_17}
  {\path{doi:10.1007/978-3-030-34029-2\_17}}.

\bibitem{edmonds_1965}
Jack Edmonds.
\newblock Paths, trees, and flowers.
\newblock {\em Canadian Journal of Mathematics}, 17:449--467, 1965.
\newblock \href {https://doi.org/10.4153/CJM-1965-045-4}
  {\path{doi:10.4153/CJM-1965-045-4}}.

\bibitem{jemalloc}
Jason Evans.
\newblock jemalloc.
\newblock \url{http://jemalloc.net/}, 2006--.
\newblock Accessed: 2025-04-15.

\bibitem{feigenbaum2005graph}
Joan Feigenbaum, Sampath Kannan, Andrew McGregor, Siddharth Suri, and Jian
  Zhang.
\newblock On graph problems in a semi-streaming model.
\newblock {\em Departmental Papers (CIS)}, page 236, 2005.
\newblock \href {https://doi.org/10.1016/j.tcs.2005.09.013}
  {\path{doi:10.1016/j.tcs.2005.09.013}}.

\bibitem{ferdous_et_al:LIPIcs.SEA.2024.12}
S.~M. Ferdous, Alex Pothen, and Mahantesh Halappanavar.
\newblock Streaming matching and edge cover in practice.
\newblock In Leo Liberti, editor, {\em 22nd International Symposium on
  Experimental Algorithms, {SEA} 2024, July 23-26, 2024, Vienna, Austria},
  volume 301 of {\em LIPIcs}, pages 12:1--12:22, Dagstuhl, Germany, 2024.
  Schloss Dagstuhl - Leibniz-Zentrum f{\"{u}}r Informatik.
\newblock \href {https://doi.org/10.4230/LIPIcs.SEA.2024.12}
  {\path{doi:10.4230/LIPIcs.SEA.2024.12}}.

\bibitem{ferdous2023streaming}
S.~M. Ferdous, Bhargav Samineni, Alex Pothen, Mahantesh Halappanavar, and Bala
  Krishnamoorthy.
\newblock Semi-streaming algorithms for weighted k-disjoint matchings.
\newblock In Timothy~M. Chan, Johannes Fischer, John Iacono, and Grzegorz
  Herman, editors, {\em 32nd Annual European Symposium on Algorithms, {ESA}
  2024, September 2-4, 2024, Royal Holloway, London, United Kingdom}, volume
  308 of {\em LIPIcs}, pages 53:1--53:19, Dagstuhl, Germany, 2024. Schloss
  Dagstuhl - Leibniz-Zentrum f{\"{u}}r Informatik.
\newblock \href {https://doi.org/10.4230/LIPIcs.ESA.2024.53}
  {\path{doi:10.4230/LIPIcs.ESA.2024.53}}.

\bibitem{froger2015set}
Aur{\'e}lien Froger, Olivier Guyon, and Eric Pinson.
\newblock A set packing approach for scheduling passenger train drivers: the
  {French} experience.
\newblock In {\em RailTokyo2015}, 2015.

\bibitem{ghaffari2017simplified}
Mohsen Ghaffari and David Wajc.
\newblock Simplified and space-optimal semi-streaming (2+epsilon)-approximate
  matching.
\newblock In Jeremy~T. Fineman and Michael Mitzenmacher, editors, {\em 2nd
  Symposium on Simplicity in Algorithms, {SOSA} 2019, January 8-9, 2019, San
  Diego, CA, {USA}}, volume~69 of {\em OASIcs}, pages 13:1--13:8, Dagstuhl,
  Germany, 2019. Schloss Dagstuhl - Leibniz-Zentrum f{\"{u}}r Informatik.
\newblock \href {https://doi.org/10.4230/OASIcs.SOSA.2019.13}
  {\path{doi:10.4230/OASIcs.SOSA.2019.13}}.

\bibitem{googlebigquerypublicdata}
{Google Cloud Platform}.
\newblock Bigquery public datasets.
\newblock \url{https://cloud.google.com/bigquery/public-data}.
\newblock Accessed: 2025-04-19.

\bibitem{DBLP:journals/corr/abs-2303-17679}
Lars Gottesb{\"{u}}ren, Tobias Heuer, Nikolai Maas, Peter Sanders, and
  Sebastian Schlag.
\newblock Scalable high-quality hypergraph partitioning.
\newblock {\em CoRR}, abs/2303.17679, 2023.
\newblock \href {http://arxiv.org/abs/2303.17679} {\path{arXiv:2303.17679}},
  \href {https://doi.org/10.48550/arXiv.2303.17679}
  {\path{doi:10.48550/arXiv.2303.17679}}.

\bibitem{gottlob2013decomposing}
Georg Gottlob and Gianluigi Greco.
\newblock Decomposing combinatorial auctions and set packing problems.
\newblock {\em J. {ACM}}, 60(4):24:1--24:39, 2013.
\newblock \href {https://doi.org/10.1145/2508028.2505987}
  {\path{doi:10.1145/2508028.2505987}}.

\bibitem{grossmann2024engineering}
Ernestine Gro{\ss}mann, Felix Joos, Henrik Reinst{\"a}dtler, and Christian
  Schulz.
\newblock Engineering hypergraph $ b $-matching algorithms.
\newblock {\em arXiv preprint arXiv:2408.06924}, 2024.
\newblock \href {https://doi.org/10.48550/arXiv.2408.06924}
  {\path{doi:10.48550/arXiv.2408.06924}}.

\bibitem{hanguir2021distributed}
Oussama Hanguir and Clifford Stein.
\newblock Distributed algorithms for matching in hypergraphs.
\newblock In Christos Kaklamanis and Asaf Levin, editors, {\em Approximation
  and Online Algorithms - 18th International Workshop, {WAOA} 2020, Virtual
  Event, September 9-10, 2020, Revised Selected Papers}, volume 12806 of {\em
  Lecture Notes in Computer Science}, pages 30--46. Springer, Springer, 2020.
\newblock \href {https://doi.org/10.1007/978-3-030-80879-2\_3}
  {\path{doi:10.1007/978-3-030-80879-2\_3}}.

\bibitem{approxresult}
Elad Hazan, Shmuel Safra, and Oded Schwartz.
\newblock On the complexity of approximating \emph{k}-set packing.
\newblock {\em Comput. Complex.}, 15(1):20--39, 2006.
\newblock \href {https://doi.org/10.1007/s00037-006-0205-6}
  {\path{doi:10.1007/s00037-006-0205-6}}.

\bibitem{huang_et_al:LIPIcs.APPROX/RANDOM.2021.14}
Chien{-}Chung Huang and Fran{\c{c}}ois Sellier.
\newblock Semi-streaming algorithms for submodular function maximization under
  b-matching constraint.
\newblock In Mary Wootters and Laura Sanit{\`{a}}, editors, {\em Approximation,
  Randomization, and Combinatorial Optimization. Algorithms and Techniques,
  {APPROX/RANDOM} 2021, August 16-18, 2021, University of Washington, Seattle,
  Washington, {USA} (Virtual Conference)}, volume 207 of {\em LIPIcs}, pages
  14:1--14:18, Dagstuhl, Germany, 2021. Schloss Dagstuhl - Leibniz-Zentrum
  f{\"{u}}r Informatik.
\newblock \href {https://doi.org/10.4230/LIPIcs.APPROX/RANDOM.2021.14}
  {\path{doi:10.4230/LIPIcs.APPROX/RANDOM.2021.14}}.

\bibitem{hurkens1989size}
Cor A.~J. Hurkens and Alexander Schrijver.
\newblock On the size of systems of sets every t of which have an {SDR}, with
  an application to the worst-case ratio of heuristics for packing problems.
\newblock {\em {SIAM} J. Discret. Math.}, 2(1):68--72, 1989.
\newblock \href {https://doi.org/10.1137/0402008} {\path{doi:10.1137/0402008}}.

\bibitem{karp2010reducibility}
Richard~M. Karp.
\newblock Reducibility among combinatorial problems.
\newblock In Michael J{\"{u}}nger, Thomas~M. Liebling, Denis Naddef, George~L.
  Nemhauser, William~R. Pulleyblank, Gerhard Reinelt, Giovanni Rinaldi, and
  Laurence~A. Wolsey, editors, {\em 50 Years of Integer Programming 1958-2008 -
  From the Early Years to the State-of-the-Art}, pages 219--241. Springer,
  2010.
\newblock \href {https://doi.org/10.1007/978-3-540-68279-0\_8}
  {\path{doi:10.1007/978-3-540-68279-0\_8}}.

\bibitem{Korte+:Greedy}
B.~Korte and D.~Hausmann.
\newblock An analysis of the greedy algorithm for independence systems.
\newblock {\em Annals of Discrete Mathematics}, 2:65--74, 1978.

\bibitem{DBLP:conf/ipps/ManneH14}
Fredrik Manne and Mahantesh Halappanavar.
\newblock New effective multithreaded matching algorithms.
\newblock In {\em 2014 {IEEE} 28th International Parallel and Distributed
  Processing Symposium}, pages 519--528. {IEEE} Computer Society, 2014.
\newblock \href {https://doi.org/10.1109/IPDPS.2014.61}
  {\path{doi:10.1109/IPDPS.2014.61}}.

\bibitem{maue2007engineering}
Jens Maue and Peter Sanders.
\newblock Engineering algorithms for approximate weighted matching.
\newblock In Camil Demetrescu, editor, {\em Experimental Algorithms, 6th
  International Workshop, {WEA} 2007, Rome, Italy, June 6-8, 2007,
  Proceedings}, volume 4525 of {\em Lecture Notes in Computer Science}, pages
  242--255. Springer, Springer, 2007.
\newblock \href {https://doi.org/10.1007/978-3-540-72845-0\_19}
  {\path{doi:10.1007/978-3-540-72845-0\_19}}.

\bibitem{mcgregor2005finding}
Andrew McGregor.
\newblock Finding graph matchings in data streams.
\newblock In Chandra Chekuri, Klaus Jansen, Jos{\'{e}} D.~P. Rolim, and Luca
  Trevisan, editors, {\em Approximation, Randomization and Combinatorial
  Optimization, Algorithms and Techniques, 8th International Workshop on
  Approximation Algorithms for Combinatorial Optimization Problems, {APPROX}
  2005 and 9th InternationalWorkshop on Randomization and Computation, {RANDOM}
  2005, Berkeley, CA, USA, August 22-24, 2005, Proceedings}, volume 3624 of
  {\em Lecture Notes in Computer Science}, pages 170--181. Springer, Springer,
  2005.
\newblock URL: \url{https://doi.org/10.1007/11538462\_15}, \href
  {https://doi.org/10.1007/11538462_15} {\path{doi:10.1007/11538462_15}}.

\bibitem{neuwohner2023passing}
Meike Neuwohner.
\newblock Passing the limits of pure local search for weighted \emph{k}-set
  packing.
\newblock In Nikhil Bansal and Viswanath Nagarajan, editors, {\em Proceedings
  of the 2023 {ACM-SIAM} Symposium on Discrete Algorithms, {SODA} 2023,
  Florence, Italy, January 22-25, 2023}, pages 1090--1137. SIAM, {SIAM}, 2023.
\newblock \href {https://doi.org/10.1137/1.9781611977554.ch41}
  {\path{doi:10.1137/1.9781611977554.ch41}}.

\bibitem{10.1145/3274668}
Ami Paz and Gregory Schwartzman.
\newblock A (2+{\(\epsilon\)})-approximation for maximum weight matching in the
  semi-streaming model.
\newblock {\em {ACM} Trans. Algorithms}, 15(2):18:1--18:15, dec 2019.
\newblock \href {https://doi.org/10.1145/3274668} {\path{doi:10.1145/3274668}}.

\bibitem{PETTIE2004271}
Seth Pettie and Peter Sanders.
\newblock A simpler linear time 2/3-$\varepsilon$ approximation for maximum
  weight matching.
\newblock {\em Information Processing Letters}, 91(6):271--276, 2004.
\newblock URL: \url{https://www.sciencedirect.com/science/article/pii/
  S0020019004001565}, \href
  {https://doi.org/https://doi.org/10.1016/j.ipl.2004.05.007}
  {\path{doi:https://doi.org/10.1016/j.ipl.2004.05.007}}.

\bibitem{Pothen+:survey}
Alex Pothen, {SM} Ferdous, and Fredrik Manne.
\newblock Approximation algorithms in combinatorial scientific computing.
\newblock {\em Acta Numerica}, 28:541--633, 2019.

\bibitem{preis1999linear}
Robert Preis.
\newblock Linear time 1/2-approximation algorithm for maximum weighted matching
  in general graphs.
\newblock In Christoph Meinel and Sophie Tison, editors, {\em {STACS} 99, 16th
  Annual Symposium on Theoretical Aspects of Computer Science, Trier, Germany,
  March 4-6, 1999, Proceedings}, volume 1563 of {\em Lecture Notes in Computer
  Science}, pages 259--269. Springer, Springer, 1999.
\newblock \href {https://doi.org/10.1007/3-540-49116-3\_24}
  {\path{doi:10.1007/3-540-49116-3\_24}}.

\bibitem{datastackexchange}
{StackExchange}.
\newblock Stackexchange data explorer.
\newblock \url{https://data.stackexchange.com/}.
\newblock Accessed: 2025-04-19.

\bibitem{trobst2024almost}
Thorben Tr{\"o}bst and Rajan Udwani.
\newblock Almost tight bounds for online hypergraph matching.
\newblock {\em arXiv preprint arXiv:2402.08775}, 2024.
\newblock \href {https://doi.org/10.1016/j.orl.2024.107143}
  {\path{doi:10.1016/j.orl.2024.107143}}.

\bibitem{dacset}
Natarajan Viswanathan, Charles~J. Alpert, Cliff C.~N. Sze, Zhuo Li, and
  Yaoguang Wei.
\newblock The {DAC} 2012 routability-driven placement contest and benchmark
  suite.
\newblock In Patrick Groeneveld, Donatella Sciuto, and Soha Hassoun, editors,
  {\em The 49th Annual Design Automation Conference 2012, {DAC} '12, San
  Francisco, CA, USA, June 3-7, 2012}, {DAC} '12, pages 774--782, New York, NY,
  USA, 2012. {ACM}.
\newblock \href {https://doi.org/10.1145/2228360.2228500}
  {\path{doi:10.1145/2228360.2228500}}.

\bibitem{williamson2011design}
David~P. Williamson and David~B. Shmoys.
\newblock {\em The {Design} of {Approximation} {Algorithms}}.
\newblock Cambridge University Press, 2011.
\newblock URL:
  \url{http://www.cambridge.org/de/knowledge/isbn/item5759340/?site\_locale=de\_DE}.

\bibitem{wolsey2020integer}
Laurence~A Wolsey.
\newblock {\em {Integer} {Programming}}.
\newblock John Wiley \& Sons, 2020.

\bibitem{wright1997primal}
Stephen~J. Wright.
\newblock {\em {Primal-Dual} {Interior-Point} {Methods}}.
\newblock {SIAM}, 1997.
\newblock \href {https://doi.org/10.1137/1.9781611971453}
  {\path{doi:10.1137/1.9781611971453}}.

\end{thebibliography}
\clearpage
\appendix
\section{Quality Guarantee for \textsf{SwapSet}}\label{sec:stackless:quality}
In the following we show that the Algorithm~\ref{alg:simple:stackless} has a quality guarantee of $\frac{1}{(1+\alpha)(\frac{d-1}{\alpha}+d)}$ for $\alpha>0$. By substituting $d=2$ and $\alpha=1$ we recover \hbox{$\frac{1}{6}$-approximation} guarantee by  Feigenbaum~et~al.~\cite{feigenbaum2005graph} on graphs. However, the best approximation guarantee is achieved when we set $\alpha=\sqrt{(d-1)/d}$. 
For graphs this returns the initial approximation guarantee of $\frac{1}{3+2\sqrt{2}}$ by McGregor~\cite{mcgregor2005finding}. 
Our analysis is inspired by McGregors~\cite{mcgregor2005finding} analysis for graphs. 
\begin{lemma}
  Let $d$ be the maximum hyperedge size  of the hypergraph $H$, $M$ be the matching returned by the \textsf{Greedy Swapping Algorithm} (Algorithm~\ref{alg:simple:stackless}), and $M^\ast$ be an optimal matching in $H$. Then
  
  \begin{align*}
    \weight(M)\geq f(\alpha,d)\cdot \weight(M^\ast)
  \end{align*}
  where $f(\alpha,d)$ is a function that depends only on $\alpha$ and $d$. Thus, 
  Algorithm~\ref{alg:simple:stackless} returns an $\frac{1}{f(\alpha,d)}$-approximate maximum weight matching.
\end{lemma}
\begin{proof}
  We refer to the set of hyperedges that are added to $M$ in the end of Algorithm~\ref{alg:simple:stackless} as \emph{survivors} and denote them by $S$.
  Now a hyperedge $e\in S$ replaced some set of hyperedges $C_1(e)$ (possibly empty), and this set of hyperedges may have replaced hyperedges in a set $C_2(e)$, and so on.
  We define $C_0(e)=e$.
  For $i\geq 1$, the set $C_i(e)$ consists of the collection of hyperedges that were replaced by hyperedges in $C_{i-1}(e)$.
   Note that a hyperedge $e'\in C_{i-1}(e)$ can be responsible for replacing at most $d$ hyperedges from the current matching.
   Let $T(e):=\bigcup_{i\geq 1}C_i(e)$ be all hyperedges that were directly or indirectly replaced by $e$. We will refer to $T(e)$ as the trail of replacement of $e$.
   
   We now show that for a hyperedge $e\in S$, the total weight of the hyperedges that were replaced by $e$ (i.e., the weight of $T(e)$) is at most $\weight(e)/\alpha$.
   \begin{claim}
    \label{claim:weight}
   For any $e\in S$, $\weight(T(e))\leq \frac{\weight(e)}{\alpha}$.
   \end{claim}
   \begin{proof}
    For each replacing hyperedge~$e$, $\weight(e)$ is at least $(1+\alpha)$ times the weight of replaced hyperedges, and a hyperedge has at most one replacing hyperedge. Hence, for all $i$, $\weight(C_i(e))\geq (1+\alpha)\cdot\weight(C_{i+1}(e))$. Thus 
    
    {\small\[
      (1+\alpha)\cdot\weight(T(e))= \sum_{i\geq 1}(1+\alpha)\cdot\weight(C_i(e))=\sum_{i\geq 0}(1+\alpha)\cdot\weight(C_{i+1}(e)) 
                                     \leq \sum_{i\geq 0}\weight(C_i)=\weight(T(e))+\weight(e).
    \]}
   Simplifying,  we obtain $\weight(T(e))\leq \frac{\weight(e)}{\alpha}$.
   \end{proof}
   Now, we consider a charging scheme to prove the Lemma. We will charge the weight of each hyperedge in $M^\ast$ to the hyperedges of $S$ and their respective trails of replacements.
   During the charging process, we will maintain an invariant: the charge assigned to any hyperedge~$e$ is at most~$(1+\alpha)\cdot\weight(e)$. 
   
   For a hyperedge $f\in S\cap M^{\ast}$, we charge $f$ to itself, which satisfies the invariant.
   Consider a hyperedge $f\in M^\ast$ that arrives in the stream but is not inserted in to $M$. 
   There are at most~$d$ conflicting hyperedges in $M$ due to which $f$ was not selected.
   We charge $\weight(f)$ to these $d$ hyperedges as follows:
   If $f$ was not chosen because of a single hyperedge~$e_1$, we assign $\weight(f)$ to $e_1$.
   For more than one hyperedge, we distribute the charge for $f$ among the hyperedges in proportion to their weights.
   Formally, for $k$ conflicting hyperedges, the charge a hyperedge $e_i$ receives is given by $\weight(f)\frac{\weight(e_i)}{\sum_{i=1}^k\weight(e_i)}$.
   Note that since $(1+\alpha)\sum_{i=1}^{k}\weight(e_i)>\weight(f)$, each of these charges  is less than $(1+\alpha)\weight(e_i)$, which satisfies the invariant.
   
   Fix a hyperedge~$e$ in the final matching $M$. Since $e$ has at most $d$ vertices, it can be charged by no more than $d$ hyperedges in an optimal solution (the so-called original charges).
   We now look at the trail of replacements.
   A hyperedge in $T(e)$ could also be charged for $d$ optimal hyperedges.
   However, for any hyperedge $t$ in $T(e)$, we can distribute some charges to $e$.
   Any charge for an optimal hyperedge $f$ that is incident at any vertex in $t\cap e$ can be directly attributed to $e$ and is included in the prior original charges.
   Therefore, the charges for any hyperedge $t$ in $T(e)$ is reduced to $d-1$.
   We now use Claim~\ref{claim:weight} to conclude 
   
   \allowdisplaybreaks
   {\small
   \begin{alignat*}{2}
    \weight(M^\ast)&\leq \sum_{e\in M}(d-1)(1+\alpha)\weight(T(e))+d(1+\alpha)\weight(e) = (1+\alpha)\sum_{e\in M}((d-1)\weight(T(e))+d\weight(e))\\
    &\overset{C.~\ref{claim:weight}}{\leq} (1+\alpha)\sum_{e\in M}\left((d-1)\frac{\weight(e)}{\alpha}+d\weight(e)\right)=(1+\alpha)\sum_{e\in M}\left(\left(\frac{d-1}{\alpha}+d\right)\weight(e)\right) \\
    &= (1+\alpha)\left(\frac{d-1}{\alpha}+d\right)\weight(M).
   \end{alignat*}
   }%
   
By setting $\alpha=\sqrt{(d-1)/d}$ the term $(1+\alpha)\left(\frac{d-1}{\alpha}+d\right)$ is minimized, resulting in an approximation ratio of $\frac{1}{(2d-1)+2\sqrt{d(d-1)}}$.
\end{proof}
\section{Instance Statistics}\label{app:stats}
\begin{table}[h]
  \caption{Statistics for the  $L_{HG}$ data set collected by Gottesbüren et al.~\cite{DBLP:journals/corr/abs-2303-17679}.
 The SPM hypergraphs are generated from %
 sparse matrices from the Suite Sparse %
 Collection~\cite{sparsecollection}.
  The hypergraphs SAT14~\cite{satbenchmark} are constructed from satisfiability clauses, and 
the DAC instances  \cite{dacset} are routability hypergraphs.
  }
  \label{tab:stats}
    {
    \centering
    \begin{tabular}{lr SSS[scientific-notation = true,round-mode = figures,
round-precision = 2,
table-format=1.1e1
]S[scientific-notation = true,round-mode = figures,
round-precision = 2,
table-format=1.1e1
]S[scientific-notation = true,round-mode = figures,
round-precision = 2,
table-format=1.1e1
]S[scientific-notation = true,round-mode = figures,
round-precision = 2,
table-format=1.1e1
]S[scientific-notation = true,round-mode = figures,
round-precision = 2,
table-format=1.1e1
]S[scientific-notation = true,round-mode = figures,
round-precision = 2,
table-format=1.1e1
]
    }\toprule
      & &\multicolumn{2}{c}{Avg. Edge Size}&\multicolumn{2}{c}{$m$}&\multicolumn{1}{c}{Edge Size}&\multicolumn{1}{c}{Node Degree}&\multicolumn{2}{c}{$n$}\\
      \cmidrule(lr){3-4}\cmidrule(lr){5-6}\cmidrule(lr){7-7}\cmidrule(lr){8-8}\cmidrule(lr){9-10}
    $L_{HG}$&\# & $\max$&$\mathrm{avg}$& $\max$&$\mathrm{avg}$& $\max$& $\max$& $\max$&$\mathrm{avg}$
      \\\midrule 
      SPM &42&140.33& 38.14& 139353211 &8958822.46&2312481&2312476&139353211&9333539.19\\
  DAC2012&10&3.69 &3.41&1340418&912788&511685&2245&1360217&924053.70\\
  SAT14&42&185.79& 12.64&53616734&14757166.07&1777221&1777221&53616734&10962158.45
  \\\bottomrule
  \end{tabular}}
  \end{table}
  \newcommand{\justprint}[1]{\num{#1}}
  \begin{table}[h]
  \caption{
  Social instances from \cite{Benson-2018-simplicial} and own processing (wiki instance, stackoverflow\_posts).
  }
  \label{tab:stats:social}
    {
    \centering\begin{tabular}{lrrrrrr
    }\toprule
       &\multicolumn{2}{c}{Edge Size}&\multicolumn{1}{c}{$m$}&\multicolumn{1}{c}{Node Degree}&\multicolumn{1}{c}{$n$}\\
      \cmidrule(lr){2-3}\cmidrule(lr){4-4}\cmidrule(lr){5-5}\cmidrule(lr){6-6}
   Social & $\max$&$\mathrm{avg}$& & \multicolumn{1}{c}{$\max$}& 
      \\\midrule 
      en-wiki-cat& \justprint{228} & 3.82 & \justprint{7847749} & \justprint{1003480} & \justprint{293177} \\
      stackoverflow\_posts & \justprint{5} & 2.96 & \justprint{11846517} & \justprint{1143864} & \justprint{45426} \\
      tags-s'flow~\cite{Benson-2018-simplicial}& \justprint{5} & 2.97 & \justprint{14458875} & \justprint{1457906} & \justprint{49998}\\
      tags-askubuntu~\cite{Benson-2018-simplicial}& \justprint{5} & 2.71 & \justprint{271233} & \justprint{21004} & \justprint{3029}\\
      tags-m'-stx~\cite{Benson-2018-simplicial}&\justprint{5} & 2.19 & \justprint{822059} & \justprint{71046}&\justprint{1629} \\
      threads-s'flow~\cite{Benson-2018-simplicial}&  \justprint{67} & 2.26 & \justprint{11305356} & \justprint{36365} & \justprint{3455072}\\
      threads-asku'~\cite{Benson-2018-simplicial}&\justprint{14} & 1.80 & \justprint{192947} & \justprint{2332} & \justprint{200974} \\
      threads-m'-stx~\cite{Benson-2018-simplicial}&  \justprint{21} & 2.24 & \justprint{719792} & \justprint{12511} &\justprint{201860} 
  \\\bottomrule
  \end{tabular}}
  \end{table}

\end{document}